\documentclass[a4paper,UKenglish,cleveref,autoref,thm-restate,pdfa]{lipics-v2021}
\hideLIPIcs
\nolinenumbers
\usepackage[textwidth=30mm]{todonotes}

\usepackage{tcolorbox}
\usepackage{tikz}

\newcommand{\sv}[1]{}
\newcommand{\lv}[1]{#1}

\newcommand{\descendant}{{\sf Desc}}

\newcommand{\bigoh}[0]{{\mathcal{O}}}

\newcommand{\types}{\ensuremath{\mathsf{Types}}\xspace}
\newcommand{\flex}{\ensuremath{\mathsf{Flex}}\xspace}

\newcommand{\height}{\mathsf{height}}
\newcommand{\winner}{\ensuremath{\alpha^*}\xspace}

\newcommand{\STF}{\ensuremath{\mathrm{STF}}\xspace}
\newcommand{\FAS}{\ensuremath{\mathrm{FAS}}\xspace}
\newcommand{\PTF}{\ensuremath{\mathrm{PTF}}\xspace}
\newcommand{\FPT}{\ensuremath{\mathrm{FPT}}\xspace}

\theoremstyle{definition}
\newtheorem{definition2}[theorem]{Definition}

\usepackage{xspace}

\usepackage{tcolorbox}
\usepackage{caption}
\usepackage{subcaption}

\newcommand{\defprob}[3]{
\begin{tcolorbox} 
  \begin{minipage}{0.99\textwidth}
  \begin{tabular*}{\textwidth}{@{\extracolsep{\fill}}ll} #1   \\ \end{tabular*}
  {\bf{Input:}} #2  \\
  {\bf{Question:}} #3
  \end{minipage}
  \end{tcolorbox}
}

\newcommand{\hide}[1]{}

\usepackage{mathrsfs}

\newcommand{\template}{blueprint\xspace}
\newcommand{\templates}{blueprints\xspace}

\title{On Controlling Knockout Tournaments Without Perfect Information}
\author{Václav Blažej}{University of Warwick, Coventry, UK}{vaclav.blazej@warwick.ac.uk}{}{}
\author{Sushmita Gupta}{The Institute of Mathematical Sciences, HBNI, Chennai}{sushmita.gupta@gmail.com}{}{}
\author{M. S. Ramanujan}{University of Warwick, Coventry, UK}{r.maadapuzhi-sridharan@warwick.ac.uk}{}{}
\author{Peter Strulo}{University of Warwick, Coventry, UK}{peter.strulo@warwick.ac.uk}{}{}
\ccsdesc{Theory of computation~Parameterized complexity and exact algorithms}
\keywords{Parameterized algorithms, tournament design, imperfect information}

\begin{document}

\maketitle

\begin{abstract}
Over the last decade, extensive research has been conducted on the algorithmic aspects of designing single-elimination (SE) tournaments. Addressing natural questions of algorithmic tractability, we identify key properties of input instances that enable the tournament designer to efficiently schedule the tournament in a way that maximizes the chances of a preferred player winning. Much of the prior algorithmic work on this topic focuses on the perfect (complete and deterministic) information scenario, especially in the context of fixed-parameter algorithm design. Our contributions constitute the first fixed-parameter tractability results applicable to more general settings of SE tournament design with potential imperfect information.


\end{abstract}


\newpage
\section{Introduction}


The algorithmic aspects of designing single-elimination (SE) knockout tournaments has been the subject of extensive study in the last decade. This format of competition, prevalent in various scenarios like sports, elections, and decision-making processes~\cite{Tullock80,DBLP:journals/ior/HorenR85,Rosen86,Laslier97,CR11}, typically involves multiple rounds, ultimately leading to a single winner. Assume that the number of players $n$ is a power of 2 and consider a complete binary tree $T$ of depth $\log n$. We can assign each player to a leaf of $T$, which gives us an initial set of pairings where  
the player at a leaf is paired with the player at the sibling of the leaf.
In the first round, the players in each pair play against each other.
Then, each loser is knocked out, we assign each winner to its parent node in $T$, and have the new siblings play against each other. Eventually there will only be one player remaining and they will be assigned to the root of the tree: they are declared the winner.
That is, in round $i$, we label the non-leaf nodes at height $i-1$ (leaves have height 0) with the winner of the match between the labels of its children. Finally, the label assigned to the root node of $T$ is the overall winner of the SE tournament.

A major research direction in this topic revolves around the algorithmic efficiency of 
computing the initial pairings in a way 
that maximizes the chances of a chosen player winning the resulting SE tournament. Here, one aims to understand the theory behind dynamics and vulnerabilities of such tournament designs from an algorithmic perspective. 

When the designer is told the winner of a match between players $i$ and $j$ for \emph{every} pair of players $i,j$ (i.e., the information is ``complete'') and these predictions are certain (i.e., the information is deterministic), we get a model with perfect information and in this case, the problem is termed Tournament Fixing (TF).
The goal of the designer is then to decide whether there is a mapping of the players to the leaves of the complete binary tree $T$ (this mapping is called a seeding) that results in the chosen player winning the resulting SE tournament. This special case of complete and deterministic information has already garnered considerable attention in the literature, particularly from the perspective of parameterized complexity starting with the work of Ramanujan and Szeider~\cite{RamanujanS17} followed by \cite{GuptaR0Z18,GuptaR0Z18a,GuptaRSZ19TFPEncoding,GuptaSS24,ChaudharyMZ24} who parameterize with {\it feedback arc set} number,  and 
\cite{Zehavi23} with {\it feedback vertex set} number.


A significant gap in the extensive literature 
that has been identified in several papers (see, for example~\cite{RamanujanS17,GuptaR0Z18a,Zehavi23,ChaudharyMZ24}), is that of the parameterized complexity of SE tournament design problems in cases where the designer {\em does not} have perfect information.
In this paper, we aim to bridge this gap
by presenting a novel parameterized algorithm for a problem we call {\sc Probabilistic Tournament Fixing} (PTF).
The input to PTF consists of an $n\times n$ matrix where (i) the $(i,j)$'th entry is denoted by $P_{i,j}$ and is some value in $[0,1]$, and (ii) $P_{i,j}=1-P_{j,i}$, see \Cref{fig:example}.
The goal is to decide if there is a seeding such that the probability of the chosen player, \winner, winning the resulting tournament is at least a given value, $p^{*}$. 

\begin{figure}[b]
        \centering
        \includegraphics[scale=0.9,page=1]{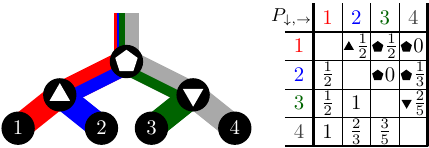}
        \caption{An example of PTF. Thickness represents probability of a player winning and probabilities that influence a match are marked by shapes. Here, player 4 wins with probability $6/10$.}%
    \label{fig:example}
\end{figure}


As well as more closely modeling the information available in real tournaments, \PTF is a generalization of TF that has applications to other perspectives. For example, tournament fixing can be seen as a special case of {agenda control}~\cite{Bartholdi89,BartholdiToveyTrick92}--a type of manipulation where a centralized authority tries to enforce an outcome.
In the backdrop of probabilistic tournament fixing, one can ask the analogous control question: Can the tournament designer ensure that a chosen player wins the tournament even when it does not have complete information about the outcome of all possible matches? 
This ``robust'' design objective can be expressed as an instance of \PTF where the target probability $p^* = 1$. Any uncertain matches can be assigned any non-integer probability: if the winner of the tournament depends on them then \winner will not win with probability 1.
This also addresses the scenario where some matches are susceptible to collusion between the players and so the designer wants to guarantee the result regardless of this adversarial behavior.

All of the aforementioned papers on tournament design problems that use parameterized algorithms are in the {\em deterministic} setting; and to the best of our knowledge, this is the first work on probabilistic tournament fixing from the perspective of parameterized complexity.

\paragraph*{Our contributions}
To obtain our parameterized algorithm for \PTF 
we introduce and solve a new problem that we call {\sc Simultaneous TF} (\STF).
Unlike \PTF, \STF is completely deterministic.
As it is easier to state our results for \STF in terms of input digraphs rather than matrices, let us first introduce the following notation.

A \emph{certainty digraph} is a graphical representation of the integral values in a given set of pairwise winning probabilities. It is a directed graph with the set of players as vertices, where the arc $ij$ exists if and only if $i$ {deterministically} beats $j$ (i.e., $P_{i,j}=1$ and $P_{j,i}=0$). A \emph{tournament digraph} is simply the certainty digraph of a matrix of pairwise winning probabilities where every possible value is present and it is either 0 or 1. Note that in a tournament digraph, there is exactly one arc between every pair of vertices, but this is not necessarily true of a certainty digraph in general. 

In \STF, one is given a sequence of tournament digraphs $D_1,\dots, D_m$ over the {same} set of $n$ players and a chosen player $\alpha^{\star}$. The goal in \STF is to compute a \emph{single} seeding (if one exists) that makes $\alpha^\star$ win in the SE tournament resulting from $D_{i}$ for \emph{every} $i$. It is straightforward to see that this problem is NP-hard. For $m=1$, it is just TF. We obtain a novel fixed-parameter (FPT) algorithm for \STF parameterized by two parameters, $k_1$ and $k_2$ where (i) the $m$ tournament digraphs agree on the orientations of all but at most $k_1$ arcs, and (ii) the largest digraph that appears as a common subgraph of every $D_i$, has a feedback arc set number (see \Cref{sec:prelims} for a formal definition) of at most $k_2$. That is, we obtain an algorithm with running time $f(k_1,k_2)n^{\bigoh(1)}$ for some function $f$. So, we get polynomial-time solvability when both $k_1$ and $k_2$ are bounded. We also show that  dropping either parameter leads to Para-NP-hardness. That is, unless P=NP, there is no FPT algorithm for STF parameterized by $k_1$ alone or by $k_2$ alone.

\textbf{Implications for PTF:}
We obtain an FPT algorithm for PTF parameterized by $c_1$ and $c_2$, where $c_1$ is the number of vertex pairs $\{i,j\}$ for which the given value of $P_{i,j}$ is non-integral, and $c_2$ is the feedback arc set number of the certainty digraph of the given set of pairwise winning probabilities. As a consequence, we get polynomial-time solvability of PTF as long as the number of fractional entries in the input is bounded {and} a natural digraph defined by the integral entries has bounded feedback arc set number.


\textbf{Implications for TF:}
Our FPT algorithm for \STF generalizes and significantly extends the known fixed-parameter tractability of TF parameterized by the feedback arc set number of the input tournament digraph ~\cite{RamanujanS17,GuptaR0Z18,GuptaRSZ19TFPEncoding}. Indeed, if we set $m=1$, implying that $k_1=0$, the value of $k_2$ is precisely this parameter. 


\paragraph*{Our motivation behind the choice of parameter}
Small feedback arc set number (FAS) is a natural condition for competitions where there is a clear-cut ranking of players with a small number of possible upsets. The relevance of this parameter is evidenced by empirical work \cite{RussellB11} and has also received significant attention in the theoretical literature on this family of problems \cite{AzizGMMSW18,RamanujanS17,GuptaR0Z18,GuptaRSZ19TFPEncoding,GuptaSS24}. Our combined parameter for PTF is in some sense a generalization of this well-motivated parameter: arcs corresponding to fractional entries ($c_1$) can have either orientation and hence can contribute to the eventual FAS, while feedback arcs in the certainty digraph are already counted by $c_2$. Note that \cite{Zehavi23} gives an FPT algorithm parameterized by feedback \emph{vertex} set number but this does not subsume our work since it is restricted to the deterministic setting.

\subsection{Related work}

The influential work of Vu et al.~\cite{DBLP:conf/atal/VuAS09} has inspired a long line of work in the topic of tournament fixing \cite{Williams10,
StantonW11,StantonW11b,KimW15,RamanujanS17,GuptaR0Z18a,GuptaR0Z18,GuptaRSZ19TFPEncoding, Zehavi23,ChaudharyMZ24,GuptaSS24} primarily in the deterministic setting. The probabilistic setting, which captures imperfect predictions of games, has also received some attention over the years, \cite{Williams10,StantonW11,StantonW11b,AzizGMMSW18}. However, the results in these papers are not algorithmic and do not necessarily hold for any given instance of PTF. Specifically, \cite{Williams10,StantonW11,StantonW11b} study the Braverman-Mossel probabilistic model for tournament generation and the existence of desirable properties such as back matchings, seedings where more than half of the top-players can win, and so on.

The topic of agenda control, initiated by Bartholdi et al.~\cite{Bartholdi89,BartholdiToveyTrick92}, within which the problems on tournament fixing lie, is of significant interest in computational social choice and algorithmic game theory. Tournament fixing as a form of agenda control has been discussed formally in~\cite[Chapter 19]{HandBook-CompSoc}.

Single-elimination tournaments have strong ties to a specific category of elections extensively explored in voting theory, namely sequential elimination voting with pairwise comparison. This is a vast area of research and we mention a few works that are most closely related to our setting. Hazon et al.~\cite{HazonDKW08} study the algorithmic aspects of rigging elections based on the shape of the voting tree and assuming probabilistic information. Additionally, Mattei et al.~\cite{MatteiGKM15} study the complexity of bribery and manipulation problems in sports tournaments with probabilistic information. Recently, \cite{CMZ24} has studied the parameterized complexity of some of these bribery questions focusing mainly on another variety of tournaments but also giving an intractability result on SE tournaments. Furthermore, Konczak et al.~\cite{KonczakL05}, Lang et al.~\cite{LangPRVW07} and Pini et al.~\cite{PiniRVW11}, study sequential elimination voting with incompletely specified preferences. Their objectives include the identification of winning candidates in either some or all complete extensions of partial preference profiles based on a given voting rule. This line of work is related to the special case of \PTF where the target probability is 1.

\section{Preliminaries}\label{sec:prelims}
We work only with directed graphs $G=(V(G),E(G))$ and refer to directed edges (i.e., arcs) as $uv \in E(G)$; note $uv \neq vu$. 
We use the notation $[n]=\{1,\dots,n\}$ and use $[a,b]$
to represent values from a range between $a$ and $b$. We make the standard assumption of dealing with inputs comprising rational numbers (see, for example, \cite{MatteiGKM15}).
Throughout the following we will fix the number of players as $n$ and $T$ as the perfect binary tree with $n$ leaves. We denote the leaves 
by $L(T)$. We will assume that $n$ is a power of 2 and that the number of levels of $T$ is $\log n$ (an integer).
We use $\descendant(v)$ to refer to the descendants of $v$ in $T$ (including $v$), i.e., the vertices $w$ such that there is a path from $v$ to $w$ away from the root. Similarly the ancestors of $v$ are the vertices $w$ with a path from $v$ to $w$ towards the root.
We define $\height(v)$ as the number of edges on a path between $v$ and its closest leaf.

A \emph{$Q$-seeding} is a function $\gamma \colon L(T) \to Q$.
If we choose $Q$ as our set of players and our seeding is bijective, this definition coincides with Definition 1 from \cite{DBLP:conf/atal/VuAS09} except that we require that $T$ is always the perfect binary tree with $n$ leaves, whereas they accept any binary tree.

\begin{definition2}[Brackets]\label{def:bracket}
     Given tournament digraph $D$ and a $V(D)$-seeding $\gamma$, a \emph{bracket generated by $\gamma$ with respect to $D$} is a labeling of $T$, defined as $\ell \colon V(T) \to V(D)$ such that  $\ell(v) = \gamma(v)$ for every leaf $v \in L$ and for every inner node $v$ of $T$ with children $u$ and $w$, $\ell(v) = \ell(u)$ if $ \ell(u)\ell(w) \in E(D)$ and $\ell(v)= \ell(w)$ otherwise. The player that labels the root of $T$ is said to {\em win} the bracket $\ell$. 
     
%
\end{definition2}

\begin{definition2}\label{def:FAS}
    A \emph{feedback arc set} (\FAS) of a digraph $D$ is a set of arcs, called \emph{back arcs}, whose reversal makes the digraph acyclic.
   The \emph{\FAS number} of $D$
   is the size of a smallest FAS.
\end{definition2}
Note that for the graph with reversed back arcs we can devise a topological vertex ordering $\prec$ such that the set of back arcs in the original graph $D$ is $\{xy \in E(D) \mid y \prec x\}$.
Note that a tournament digraph is acyclic if and only if it is also transitive.

In proofs, we use the following technical folklore tool.
\begin{lemma} \label{lem:lognfpt}
    For every $k,n \in \mathbf{N}$ where $k \leq n$, we have $(\log n)^k \leq (4k \log k )^k + n^2$
\end{lemma}

\newcommand{\ilpfeas}{{\sc ILP-Feasibility}\xspace}

We will use as a subroutine the well-known FPT algorithm for \ilpfeas.
The \ilpfeas problem is defined as follows. The input is a  matrix $A\in {\mathbb Z}^{m\times p}$ and a vector $b\in {\mathbb Z}^{m \times 1}$ and the objective is to find a vector $\bar x\in {\mathbb Z}^{p \times 1}$ satisfying the $m$ inequalities given by $A$, that is, $A\cdot \bar x\leq b$, or decide that such a vector does not exist.  

\begin{proposition}[\cite{Lenstra83,Kannan87,FrankTardos87}]\label{prop:lenstra}
    \ilpfeas can be solved 
    using $\bigoh(p^{2.5p+o(p)}\cdot L)$ arithmetic
     operations and space polynomial in $L$, 
    where $L$ is the number of bits in the input and $p$ is the number of variables.
\end{proposition}

\section{The Algorithm for \STF and its Analysis}

In this section we will present the algorithm for \STF. In \Cref{sec:applications} we will show the application to \PTF.
We remark that a reader interested in the application can skip ahead to this section immediately after parsing the main statement (Theorem \ref{thm:SimulTFPisFPT}).

Recall that the {\STF} problem is formally defined as follows.

\defprob{\textsc{Simultaneous Tournament Fixing (\STF)}}{
   A sequence of tournament digraphs $D_1,\dots,D_m$ on vertex set $N$ and a player $\winner \in N$.
}{
   Does there exist an $N$-seeding $\gamma$ such that for all $i \in [m]$, $\winner$ wins the bracket generated by $\gamma$ with respect to $D_i$?
}

Given tournament digraphs $D_1,\dots,D_m$ on common vertex set $N$, we define the set of \emph{shared} arcs, $\widehat{E} = \bigcap_{i=1}^m E(D_i)$. The remaining vertex pairs are the \emph{private} arcs. The \emph{shared digraph} is the graph $(N,\widehat{E})$.
We call the FAS of the shared digraph the \emph{shared FAS} and denote it $\widehat{F}$ and let $\prec$ denote the ordering where the back arcs of the shared FAS is the set $\{xy \in \widehat{E} \mid y \prec x\}$.
This ordering is typically depicted left-to-right so for $y \prec x$ we also say $y$ is \emph{left} of $x$ and $x$ is \emph{right} of $y$.

\begin{theorem}\label{thm:SimulTFPisFPT}
   \STF is \FPT parameterized by the size of the shared \FAS and the number of private arcs.
\end{theorem}

We argue that {\em both} parameters in the above statement are required, by showing that \STF is NP-hard even when either one of the parameters is a constant (i.e., \STF is para-NP-hard parameterized by either parameter alone). 

\begin{lemma}\label{lem:minimalParameterizationSTF}
\STF parameterized by the shared \FAS is para-NP-hard and 
\STF parameterized by the number of private arcs is para-NP-hard.
\end{lemma}

\begin{proof}
    Notice that if $m=1$, then the number of private arcs is 0 and we get the TF problem, which is NP-hard. On the other hand, let $(D,\winner)$ be an instance of TF. Construct an instance $(D_1,D_2,\winner)$ of {\STF} by setting $D_1=D$ and defining $D_2$ as the tournament obtained by taking an arbitrary acyclic (re-)orientation of the arcs of $D$, such that all arcs incident on $\winner$ are oriented away from it (i.e., $\winner$ beats everyone else).
    Then, it is easy to see that the constructed instance of {\STF} is a yes-instance if and only if the original TF instance is also a yes-instance. Moreover, the shared digraph in the {\STF} instance is a subgraph of $D_2$ and so is acyclic, i.e., the instance has a shared FAS of size 0.
\end{proof}

The rest of the section is devoted to the proof of Theorem \ref{thm:SimulTFPisFPT}. 
For ease of description, in the rest of this section we will denote our combined parameter (i.e., the sum of the size of the shared FAS and number of private arcs) by $k$.
Note that we can assume that the given tournament digraphs are distinct and that $m \leq 2^k$.
This holds because $k$ upper bounds the number of private arcs and there are two possibilities of how a private arc may be present in a tournament digraph (either $uv$ or $vu$) so having more than $2^k$ necessarily produces duplicate tournament digraphs.


Our parameterization allows us to define a small number of ``interesting'' vertices as follows.
We say a vertex is \emph{affected} if it is an endpoint of either a feedback arc or a private arc. Additionally, $\winner$ is also affected. More precisely, the set of affected vertices is $V(\widehat{F}) \cup V(E(D_1) \setminus \widehat{E}) \cup \{\winner\}$ denoted by $V_A = \{a_1, \dots, a_{k'}\}$ where the ordering agrees with $\prec$ (i.e. $a_j \prec a_{j+1}$). Note that $|V_A|\leq 2k+1$.

We now use these affected vertices as ``breakpoints'' and classify the remaining vertices by where in the order they fall relative to the affected vertices.
Define a set $\types = [|V_A|+1] \cup V_A$ and a \emph{type function} $\tau \colon N \to \types$ where $\tau(v) = v$ for each $v \in V_A$ and otherwise $\tau(v) = j$, where $j$ is the smallest index in $[|V_A|]$ such that $v \prec a_j$. If there is no such index set $\tau(v) = |V_A| + 1$. We say a vertex $v$ is of type $t$ if $\tau(v)=t$. We refer to types in $\flex = \types \setminus V_A$ as \emph{flexible} types and others as \emph{singular} types.
See \Cref{fig:types} for an example.

\begin{figure}[b]
   \centering
   \includegraphics[scale=1.0,page=1]{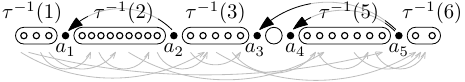}
   \caption{
        The set $\types = (1,a_1,2,a_2,3,a_3,4,a_4,5,a_5,6)$ sorted according to $\prec$.
        The flexible types are depicted containing the players of that type; note $\tau^{-1}(4)=\emptyset$.
        The back arcs are depicted above while the remaining arcs go ``right'', below the vertices we depict some of the remaining arcs for illustration.
    }%
    \label{fig:types}
\end{figure}

This ``classification'' of players is motivated by the following observation, which implies that for any three distinct players, two of which are of the same type, and the third is of a different type, the first two have the exact same win/loss relationship with the third one.

\begin{observation}[slightly extended Lemma 2 in \cite{RamanujanS17}]\label{obs:types_behave_uniformly_wrt_others}
    For every $i \in [m]$, type $t \in \types$, distinct players $u,v \in \tau^{-1}(t)$, and $w \in N$ such that $\tau(w) \ne t$, $wu \in E(D_i)$ if and only if $wv \in E(D_i)$.
\end{observation} 

The above observation allows us to construct a tournament digraph on $\types$. Moreover, in our search for a seeding which makes $\winner$ win, \Cref{obs:permute_typed_players} implies that we can ignore the specific placement of the players that have the same type in relation to each other. 
\begin{observation}\label{obs:permute_typed_players}
    If two seedings $\gamma_1$ and $\gamma_2$ of the same tournament digraph $D$ differ only in placement of players that have the same type, i.e., $\gamma_1(u) \ne \gamma_2(u) \implies \tau(\gamma_1(u)) = \tau(\gamma_2(u))$, then $\winner$ wins the bracket generated by $\gamma_1$ with respect to $D$ if and only if $\winner$ wins the bracket generated by $\gamma_2$ with respect to $D$.
\end{observation}
This allows us to view the solution to \STF as a seeding on the set \types and fill in the actual players at a later stage. This insight leads us to the following definitions.


\paragraph*{Digraph, seeding and bracket of \types.}We construct a tournament digraph, denoted by $D_i^\tau$, with vertex set $\types$ where type $x$ beats type $y$ in $D_i^\tau$ when vertices of type $x$ beat vertices of type $y$ in $D_i$. Precisely stated, $D_i^\tau = (\types, E)$ where $xy \in E$ if and only if there exist $u,v \in N$ such that $\tau(u)=x$, $\tau(v)=y$, and $uv \in E(D_i)$. We call $D_i^\tau$ the \emph{\types-digraph generated by $D_i$}.
Since $D_i^\tau$ is a tournament digraph we can define a topological ordering $\prec$ in much the same way as for the original digraph. Since the flexible types are not affected vertices there are no back arcs so we have the following observation.
\begin{observation}
    For all $x, y \in \flex$, $x \prec y$ if and only if $xy \in E(D_i^\tau)$ (i.e. $x$ beats $y$).
\end{observation}

As discussed earlier, from now on we will use a \types-seeding to represent the solution and fill in the actual players in place of their types at the very end. Given an $N$-seeding $\gamma$, we define the \types-seeding $\beta \colon L \to \types$ as $\beta(u) = \tau(\gamma(u))$, for each $u \in L(T)$. 
This corresponds to taking the seeded players and mapping them to their types.
Finally, for each $i \in [m]$, we define $\ell_i$ as the bracket generated by $\beta$ with respect to $D_i^\tau$.


To exploit \Cref{obs:permute_typed_players} we focus on bracket vertices that may be influenced by differences between the input tournament digraphs.
Bracket vertices that depend only on the shared digraph always have the same winner (\Cref{obs:non_template_equal}), but bracket vertices that have affected vertices as their descendants may have different winners in different input tournaments.
This distinction inspires definition of the following small structure that we can guess (by iterating through every possibility) and later extend to find a solution.


\begin{definition2}\label{def:template} 
    For any $N$-seeding $\gamma$,
    the \emph{\template} generated by $\gamma$ with respect to $D_1, \dots, D_{m}$,  denoted by $\mathcal{T}$, consists of
    \begin{itemize}
        \item a \template subtree $T' \subseteq T$. This is the subtree of $T$ induced by the ancestors of leaves $\{\gamma^{-1}(a) \mid a \in V_A\}$, and
        \item $m$ labelings of\/ $T'$ constructed by restricting $\ell_i$ to $V(T')$ for each $i \in [m]$.
    \end{itemize} 
\end{definition2}
We will abuse the notation in the context of \templates, by using $\ell_i$ to denote the labelings restricted to $T'$. See \Cref{fig:template} for an example of a \template.


Note that for an instance of \STF there are a total of $n!/2^{n-1}$ possible choices for the solution seeding. Our algorithm approaches this massive search space by instead finding a {\template} that preserves enough information from the full solution but is simple enough that there are not many of them. Several papers, such as \cite{RamanujanS17,GuptaR0Z18,GuptaR0Z18a,GuptaRSZ19TFPEncoding}, use a similar approach where a substructure called a {\em template} is used. However, our approach differs significantly from these papers since our notion of a \template is based on the bracket and hence a complete binary tree representation of the outcome of the tournament. On the other hand, the previous papers working with \FPT algorithms for TF have generally defined so-called templates using the notion of {\em spanning binomial arborescences} (SBA), which are a specific type of spanning trees. Indeed, TF has a well-established connection to SBAs \cite{Williams10}, that states that there is a solution to the TF instance if and only if the tournament digraph has an SBA rooted at the favorite player.
However, SBAs are unsuitable for our setting since they directly represent the winner of each match in their structure, in our case we would have to deal with multiple SBAs in parallel, which appears to be challenging, technically. 
As a result, the notion of \template considered here, \Cref{def:template}, is a significant deviation from the literature on FPT algorithms for this type of problems.

\begin{figure}[t]
   \centering
   \includegraphics[scale=1.0,page=1]{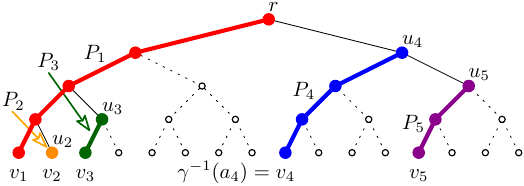}
   \caption{
        Example of a \template subtree on a bracket with $n=16$ and $|V_A|=5$.
        The \template subtree consists of solid vertices and edges.
        Empty vertices and dotted edges are in the bracket, but not in the \template.
        The colored thick edges mark the \template partition into paths $P_j$ for $j \in [5]$.
    }%
    \label{fig:template}
\end{figure}

\begin{observation}\label{obs:non_template_equal}
   Every vertex that is not in the \template subtree has the same label in every $\ell_i$.
   More precisely, for every $u \in V(T) \setminus V(T')$ we have $\ell_i(u) = \ell_j(u)$ for all $i, j \in [m]$.
\end{observation}

That is, the \template includes all of the information that changes between the different input tournament digraphs.
We want to find \templates without knowing the seeding that generates them. To do this we use the following lemma which specifies conditions under which $(T', \ell_1, \dots, \ell_m)$ constitutes a \template.

\begin{lemma}\label{lem:properties}
    Given $m$ \types-digraphs $F_i$, a subtree $T' \subseteq T$, and $m$ labelings $\ell_i \colon V(T') \to \types$, suppose:
    \begin{enumerate}
        \item\label{case:twochildren} for each $v \in V(T')$ with two children in $T'$, $u$ and $w$, for each $i \in [m]$ we have
        $\ell_i(v) = \ell_i(u)$ if $ \ell_i(u)\ell_i(w) \in E(F_i)$ and $\ell_i(v) = \ell_i(w)$ otherwise.
        \item\label{case:onechild} for each $v \in V(T')$ with exactly one child $u$ in $T'$ there exists a type $t \in \flex$ such that for every $i \in [m]$ let $q:=\ell_i(u)$, either $qt \in E(F_i)$ and $\ell_i(v) = q$ or $tq \in E(F_i)$ and $\ell_i(v) = t$, and
        \item\label{case:nochild} every vertex $v$ with no children in $T'$ is also a leaf in $T$ and $\ell_i(v)=\ell_j(v) \in V_A$ for all $i,j \in [m]$.
    \end{enumerate}
    Then there exist $m$ tournament digraphs, $D_1,\dots,D_m$ and an $N$-seeding $\gamma$ such that $(T', \ell_1,\dots,\ell_m)$ is the \template generated by $\gamma$ with respect to $D_1,\dots,D_m$ and $F_i$ is the \types-digraph generated by $D_i$ for each $i \in [m]$.
\end{lemma}
\begin{proof}
    We aim to construct a seeding $\gamma$ and tournament digraphs $D_1,\dots,D_m$ which generate our \template $(T',\ell_1,\dots,\ell_m)$.
    Due to \Cref{def:template} and Condition~\ref{case:nochild} we know that $\gamma(v) = \ell_i(v)$ for all $v \in L(T')$ and $i \in [m]$, which fixes all affected vertices of the tournament digraphs.
    As the remaining leaves $Q=L(T)\setminus L(T')$ are not part of the \template, we know by \Cref{obs:non_template_equal} that $\ell_i(v)=\ell_j(v)$ for all $i,j\in [m]$ and $v \in Q$.
    All singular types were assigned, hence, the leaves $Q$ must be assigned players with flexible types.
    In particular, for every non-\template vertex $w$ whose parent $v$ is in the \template we can find a type $t \in \flex$ according to Condition~\ref{case:onechild}.
    To ensure that vertex $w$ indeed gets type $t$ we take the set of all the leaves under $w$ and for each of them $q \in Q \cap \descendant(w)$ we create a new player $p$ with type $\tau(p)=t$ and assign $\gamma(q)=t$.
    We see that this way each leaf $q \in Q$ gets a new player because when we follow a path from the leaf to the root then the first vertex of $T'$ we encounter is $v$ and the one before is $w$ which when processed as described above assigned a player to $q$.
    Hence, $\gamma$ is complete.
    Finally, for each $i \in [m]$ as $F_i = D_i^\tau$ we know that the singular types of $F_i$ directly translate to players in $D_i$.
    The players that get flexible types were created when we devised $\gamma$.
    As the last step, we add arcs to $D_i$ so that $uv \in E(D_i) \iff \tau(u)\tau(v) \in E(F_i)$ as implied by \Cref{obs:types_behave_uniformly_wrt_others}.
\end{proof}

We want to find a bound on the number of \templates so we can enumerate them in the desired time complexity.
Towards that, we first show that the sequences of labels along paths from leaves to the root are well structured.

\begin{lemma}\label{lem:bound_path_changes}
   Let $s_1,s_2,\dots,s_{\log n}$ be a sequence of types assigned by $\ell_i$ along a path in $T$ starting in a leaf $x$ up to the root $r$, i.e., $s_1 = \ell_i(x)$ and $s_{\log n} = \ell_i(r)$, then the number of positions $j \in [\log n-1]$ such that $s_j \ne s_{j+1}$ is $2k(k+1)$. 
\end{lemma}
\lv{
\begin{proof}
   We take a tournament digraph $D_i$ and observe, that every of its arcs is used at most once because such match implies one of its players was eliminated from the bracket.
   In our case, however, some types ($\tau$) represent multiple vertices so arcs in $D_i^\tau$ may be repeated if one of their endpoints is not a singleton type.
   We have $\prec$ the ordering of the players that witnesses that the shared graph has FAS at most $k$.
   Consider one $i \in [\log n-1]$ such that $s_i \ne s_{i+1}$.
   These types represent players $x$ and $y$ such that $\tau(x)=s_{i+1}$ and $\tau(y)=s_i$ and $x$ won against $y$ so we have arc $xy \in E(D)$.
   There are two possibilities, either the new type $s_{i+1}$ is strictly right of $s_i$ or strictly left of $s_i$ with respect to the ordering $\prec$.
   In the strictly right case, say $s_i$ is $j$-th type when we order the types from $1$ up to $2k+1$ in the $\prec$ ordering.
   As $s_i \prec s_{i+1}$ we know that $s_{i+1}$ is at least $(j+1)$-th in the $\prec$ ordering.
   Looking at the whole sequence of type labels, the case where $s_i \prec s_{i+1}$ may repeat at most $2k$ times in a row.
   In the strictly left case, we have $s_{i+1} \prec s_i$ so $xy$ is a back arc so this case may appear at most $k$ times in total.
   For each left case we may have the right case repeating up to $2k$ times which gives us an upper bound on the number of positions where the type changes of $2k(k+1) \in \bigoh(k^2)$.
\end{proof}
}

While there are $\bigoh(n^n)$ $N$-seedings and even $\bigoh(k^n)$ $\types$-seedings, there are only $f(k) \cdot n^{\bigoh(1)}$ different \templates.
\begin{lemma}\label{lem:number-of-templates}
   There exists a function $f$ such that there are only $f(k) \cdot n^{\bigoh(1)}$ \templates.
   Additionally, it is possible to iterate through all of them that agree with a given sequence of \types-digraphs in \FPT time.
\end{lemma}

\sv{
We only include a proof sketch here, the full proof can be found in the appendix. 

\begin{proof}[Proof Sketch]
    We view the \template tree $T'$ as a collection of paths from the affected vertices to the root. We can guess $T'$ by simply guessing the length of each path before it intersects with another path and the index of the path that it first intersects with.
    Since the labels of each path only change $\bigoh(k^2)$ times we can guess the length and label of each ``run'' of labels that are the same.
    We can use this process to guess each of the $m$ labelings.
    Each guess is over at most $(\log n)^{f(k)}$ possibilities hence the overall runtime is FPT by \Cref{lem:lognfpt}.
\end{proof}
}
\lv{
\begin{proof}
    First, we partition the \template subtree $T'$ into $|V_A|$ paths $P_j$ in the following way, see \Cref{fig:template}.
    Each path $P_j$ has one endpoint in a leaf $v_j$ that maps to an affected vertex $\gamma(v_j) = a_j$, $P_1$ has the other endpoint in root $r$, and all the other paths $P_j$ for $j \geq 2$ end in a vertex $u_j$ whose parent belongs to a different path $P_{j'}$, $j' \ne j$.
    These decompositions are characterized having for each $2 \leq j \leq |V_A|$ by height of $u_j$ (i.e. length of $P_j$) and index of the path parent of $u_j$ belongs to.
    Height of $u_j$ is upper bounded by $\log n$ and we can upper bound possible indices of parent paths by $|V_A|$.
    This gives us an upper bound on the number of \template subtrees $(\log n \cdot |V_A|)^{|V_A|}$.
    Note we may be overcounting as one may decompose $T'$ in multiple ways, however, decomposition gives a unique way to retrieve $T'$.
    So every possible $T'$ is counted and the total number of \templates is upper bounded by the number of decompositions.

    Second, we need to upper bound the number of possible \template labelings $\ell_i$ for $i \in [m]$.
    Note that each $P_j$ of the \template subtree decomposition is part of path $P'_j$ that goes from $v_j$ to the root $r$.
    Due to \Cref{lem:bound_path_changes} we can upper bound the number of changes of $\ell_i$ along $P'_j$ to $\bigoh(k^2)$.
    Joining this bound over all $i \in [m]$ we get that there are no more than $\bigoh(m \cdot k^2)$ changes (in any labeling) along $P'_j$.
    Hence, we can represent labelings of $P'_j$ by $\bigoh(m \cdot k^2)$ runs where each run is a tuple made of a run label and a run integer.
    Run labels have $(2k+1)^m$ possible values because they combine labels $\ell_i$ for all $i \in [m]$.
    Run integers say how many vertices have the specified labels in a row and are in $[\log n]$.
    Therefore, there are at most $\big((2k+1)^m \cdot \log n\big)^{\bigoh(m \cdot k^2)}$ labelings of $P'_j$.
    Repeating this argument again for all paths $P'_j$, $j \in |V_A|$, we get that the \template subtree contains no more than $\big((2k+1)^m \cdot \log n\big)^{\bigoh(m \cdot k^2 \cdot |V_A|)}$ vertices $u$ that have a child $v$ such that $\ell_i(u) \ne \ell_i(v)$ for some $i \in [m]$.

    Combining the above bounds we have at most $(\log n \cdot |V_A|)^{|V_A|} \cdot \big((2k+1)^m \cdot \log n\big)^{\bigoh(m \cdot k^2 \cdot |V_A|)}$ labeled \templates.
    Finally, by the bounds $m \leq 2^k$ and $|V_A| \leq 2k+1$ we get $(k^{2^k} \cdot \log n)^{\bigoh({2^k} \cdot k^3)}$.
    Choosing $f(k) = k^{2^{k^2}}$ upper bounds the total number of possible \templates by $f(k) \cdot n^{\bigoh(1)}$ by \Cref{lem:lognfpt}.

    To iterate through all the \templates we 
    first guess the parent of each path $P_j$ and the length of each path before it reaches its parent path for $j \geq 2$.
    This gives a tree $T'$. We have $T' \subseteq T$ if and only if $T'$ is binary so we check that every vertex has at most two children.
    Next we guess the labels for each $P_j'$ by guessing $m \cdot 2k(k+1)$ runs. We then check that the guessed labels agree on every vertex that is shared between the $P_j'$s.
    Finally we check the conditions of \Cref{lem:properties}. This tells us that our guess is a \template generated by $\gamma$. We have already shown that the number of guesses we could make is upper bounded by $f(k) \cdot n^{\bigoh(1)}$.
\end{proof}
}


This shows we can guess the \template efficiently. Next, we take a \template $\mathcal T$ and construct an ILP where the existence of a feasible solution is both a necessary and sufficient condition for the existence of a solution to the instance of \STF with \template $\mathcal T$.


In the remainder of the paper we frequently use the following set of specific \template vertices. We also divide it into two sub-categories.
\begin{definition2}[Important vertices]\label{def:important_vertices}
      For a \template $\mathcal T=(T',\ell_1,\dots,\ell_m)$ the set of \emph{important vertices} $\mathrm{imp}(\mathcal T)$ is the subset of vertices in $V(T) \setminus V(T')$ that have a parent from $V(T')$.
    I.e., these are the non-\template children of the \template vertices; note that a \template vertex cannot have two non-\template children.
    

    For each important vertex $w \in \mathrm{imp}(\mathcal T)$, let $v$ be the parent of $w$, and let $u$ be the sibling of $w$ (the other child of $v$). Note that $T$ and $T'$ agree on these relationships and $u,v \in V(T')$ while $w$ is only in $V(T)$. If $\ell_i(v) = \ell_i(u)$ for all $i \in [m]$, we add $(u,v,w)$ to $J_{\mathcal{T}}$. Otherwise there exists $i \in [m]$ such that $\ell_i(v) \neq \ell_i(u)$ and we add $(u,v,w,i)$ to $K_{\mathcal{T}}$ for some such $i$.
\end{definition2}

The reason for these two categories is that wherever the type label changes at a vertex with only one child in the \template tree, the other child must be labeled with the new type (\Cref{lem:properties}). This reduces the number of players needed to pack into this subtree by one.


\paragraph*{ILP Feasibility.}\label{ILP}

Given a \template $\mathcal{T} = (T', \ell_1, \dots, \ell_m)$, we initialize variables $b_t$ and $c_t$ for $t \in \flex$ where $b_t$ keeps the number of leaves that need to be mapped to a type that either is equal to $t$ or is beaten by $t$ and $c_t$ keeps the number of players of type $t$ that remain to be assigned to the solution.
We initially set $b_t = 0$ and $c_t = |\{j \in N : \tau(j)=t\}|$ for all $t \in \flex$.


We now consider each important vertex by iterating through $J_{\mathcal{T}}$ and $K_{\mathcal{T}}$ (in any order):
\begin{itemize}
        \item For each $(u,v,w) \in J_{\mathcal{T}}$
        we add $2^{\height(w)}$ to $b_q$ where $q$ is the strongest type from \flex that gets beaten by all $\ell_1(v),\dots,\ell_m(v)$.
        \item For each $(u,v,w,i) \in K_{\mathcal{T}}$
        we add $2^{\height(w)}-1$ to $b_{\ell_i(v)}$ and then decrement $c_{\ell_i(v)}$ by one.
        Recall that $i$ in this case reflects $\ell_i(v) \ne \ell_i(u)$.
\end{itemize}

\begin{figure}[b]
    \centering
    \begin{subfigure}[b]{0.2\textwidth}
        \centering
        \includegraphics[scale=1.0,page=1]{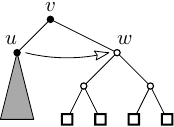}
        \caption{$\forall i : \ell_i(v)=\ell_i(u)$}%
        \label{fig:case_a}
    \end{subfigure}
    \qquad\qquad\qquad
    \begin{subfigure}[b]{0.2\textwidth}
        \centering
        \includegraphics[scale=1.0,page=2]{figures/cases.pdf}
        \caption{$\exists i : \ell_i(v)\ne \ell_i(u)$}%
        \label{fig:case_b}
    \end{subfigure}
    \caption{
        Depiction of $\ell_i$ in cases during creation of the ILP instance.
        Arrows depict match where the player at arrow tail wins.
        Red bold path depicts the player that gets to $v$ in $\ell_i$.
        \textbf{(\subref{fig:case_a})} $u$ beats $w$ so all leaves necessarily lose to $u$.
        \textbf{(\subref{fig:case_b})} There is a leaf $z$ that beats $u$ and it also beats all other leaves.
    }%
    \label{fig:cases}
\end{figure}


If $c_t < 0$ for any $t \in \flex$ then the \template requires more players of type $t$ than are available, so we can safely reject it.


Now, we define our instance of ILP-Feasibility, denoted by $I_{\mathcal{T}}$, over a set of $\bigoh(k^{m+1})$ non-negative constants $b_s$ and $c_t$, and $\bigoh(k^{m+1})$ variables $x_{s,t}$, where $s, t \in \flex$.
The ILP has the following constraints.
\begin{align}
    \text{For all } t\in \flex, \quad &\sum_{s \in \flex}{x_{s,t}} = c_t, \label{const:flex}\\
    \text{For all } s\in \flex, \quad &\sum_{t \in \flex}{x_{s,t}} = b_s,  \label{const:flex-tuples}\\
    \text{For all } s, t\in \flex \text{ with } t \prec s \quad &x_{s,t}=0, \label{const:cannot-beat}\\
    \text{For all } s, t \in \flex \quad &x_{s,t} \geq 0. \label{const:nonnegative}
\end{align}

The intuition is as follows.
Variable $x_{s,t}$ represents how many leaves we assigned to have type $t$ that cannot be assigned a type stronger than $s$.
\Cref{const:flex} ensures the number of players assigned type $t$ is equal to the number of players we have available,
\Cref{const:flex-tuples} ensures that each group of leaves with strongest type $s$ is assigned the correct number of players,
\Cref{const:cannot-beat} forbids assignment of a player that is too strong to a group,
and we have \Cref{const:nonnegative} so that we are assigning a non-negative number of players.
 
\hide{
For each type $t \in \flex$ and $m$-tuple of types $s \in \flex^m$ we have a variable $x_{s,t}$. We now describe the constraints.
\begin{enumerate}
    \item For each $t \in \flex$ we have the constraint $\sum_{s \in \flex^m}{x_{s,t}} = c_t$.
    \item For each $s \in \flex^m$ we have the constraint $\sum_{t \in \flex}{x_{s,t}} = b_s$.
    \item For each $s \in \flex^m$ and $t \in \flex$ where there exists $i \in [m]$ such that $t \prec s_i$ we have the constraint $x_{s,t}=0$.
\end{enumerate}
}

\begin{lemma}\label{lem:feasible_stf_has_ilp}
    Given an instance of \STF that has a solution $\gamma$, there exists a \template, $\mathcal{T}$, such that $I_{\mathcal{T}}$ is feasible.
\end{lemma}
\begin{proof}
    We construct the \template generated by $\gamma$, and denote it $\mathcal{T} = (T', \ell_1, \dots, \ell_m)$.
    We will construct a solution to $I_{\mathcal{T}}$ as follows.
    
    Initially set $x_{s,t} = 0$.
    Then we process the important vertices, see \Cref{def:important_vertices}, as follows.
    \begin{itemize}
        \item For each $(u, v, w) \in J_{\mathcal{T}}$
            we have that $\ell_i(u)$ beats $\ell_i(w)$ for all $i \in [m]$.
            So for each $t \in \flex$ we add $|\{y \in \descendant(w) \cap L(T) : \tau(\gamma(y))=t\}|$ to $x_{q, t}$ where $q$ is the strongest flexible type that gets beaten by all $\ell_1(v),\dots,\ell_m(v)$.
        
        \item For each $(u, v, w, i) \in K_{\mathcal{T}}$
            we have that $\ell_i(w)$ beats $\ell_i(u)$.
            Even if this also occurs for some other $i$ there is still only one new label by \Cref{obs:non_template_equal} (the label of $u$ could be different but there is only one label of $w$).
            We can trace back the type $\ell_i(w)$ to a leaf $z \in \descendant(w) \cap L(T)$ where $\tau(\gamma(z)) = \ell_i(w)$.
            For each $t \in \flex$ we add $|\{y \in \descendant(w) \cap L(T) \setminus \{z\} : \tau(\gamma(y))=t\}|$ to $x_{\ell_i(w), t}$.

        
    \end{itemize}

    Effectively, we are taking the important vertices as shown in \Cref{def:important_vertices} and assigning the values $x_{s,t}$ according to the seeding $\gamma$.
    
    However, we also know that at least one leaf that is a descendant of a vertex where the label in the \template changes must be assigned a player with the type that the label changes to. We exclude these known leaves from the count.
    It remains to show that this assignment satisfies all the constraints.

    We count each player exactly once except for $\gamma(z)$ for tuples in $K_{\mathcal{T}}$. Since $c_t$ starts as the total number of players of type $t$ and the appropriate $c_t$ is decremented for each tuple in $K_{\mathcal{T}}$,
    the total number of players of type $t$ we count is exactly $c_t$
    so \cref{const:flex} is satisfied.

    Every leaf of $T$ that is not in the \template is a descendant of exactly one important vertex.
    Hence it counts towards exactly one $b_s$ except for $z$ for tuples in $K_{\mathcal{T}}$ which is excluded. 
    Therefore the total number of leaves (regardless of the type of the player assigned to them by $\gamma$) that are descendants of $w$ is $2^{\height(w)}$ for tuples in $J_{\mathcal{T}}$ and one less for tuples in $K_{\mathcal{T}}$. Summing over all possible important vertices gives us $b_s$, so \hide{the second group of constraints }\cref{const:flex-tuples} is satisfied.
    
    We know that $\ell_i(x) \prec \ell_i(v)$ never happens for any $i \in [m]$ and $x \in \descendant(w)$ by the definition of bracket and since every back arc is between affected vertices and $x$ cannot be an affected vertex. In particular the flexible types of the leaves $y \in \descendant(w) \cap L(T)$ cannot beat $\ell_i(v)$ for any $i \in [m]$. This means they will never contribute to any $x_{s,t}$ where $t \prec s$ and hence the assignment of $x_{s,t}$ will satisfy \hide{the third group of constraints}\cref{const:cannot-beat}. 
\end{proof}

We now show the converse of \Cref{lem:feasible_stf_has_ilp}.
\begin{lemma}\label{lem:feasible_ILP_has_STF}
    If there exists a \template $\mathcal{T}$ where $I_{\mathcal{T}}$ has a solution, then the instance of\/ \STF is a yes-instance.
\end{lemma}

\begin{proof}
    We construct the solution $\gamma$ by assigning players to $\gamma(v)$ for all $v \in L(T)$.
    For this proof, we define \emph{assign player $p$ to leaf $v$} as setting $\gamma(v)=x$, and marking $x$ and $v$ so that they may not be assigned later.
    The marking plays a role when we \emph{assign set $P$ to set $L$} which means we assign $|L|$ arbitrary players of $P$ to arbitrary leaves in $L$ one by one.

    We start by assigning to the leaves of the \template.
    The labelings in any \template always agree on the leaves, i.e., for $v \in L(T')$ we have $\ell_i(v)=\ell_j(v)$ for all $i,j \in [m]$.
    So we can assign $\ell_i(v)$ to $v$ for every $v \in L(T')$.

    It remains to assign players to the other leaves of $T$.
    Note that all the remaining players have flexible types and they form a total order.
    Let $B_s$ where $s \in \flex$ be called a \emph{bag}, we will gradually add leaves to the bags.
    Each bag indicates what is the strongest flexible type $s$ its leaves can be assigned to.
    Initially, we set all bags $B_s = \emptyset$.
    Our aim is to add the remaining leaves to the bags in a way that $|B_s|=b_s$.
    We process each important vertex as follows:
    \begin{itemize}
        \item For each $(u, v, w) \in J_{\mathcal{T}}$ we
            add $L(T) \cap \descendant(w)$ to $B_q$ where $q$ is the strongest flexible type that gets beaten by all $\ell_1(v),\dots,\ell_m(v)$.
        \item For each $(u, v, w, i) \in K_{\mathcal{T}}$ we
            choose an arbitrary $z \in L(T) \cap \descendant(w)$ and $j \in \tau^{-1}(\ell_i(v))$,
            assign $j$ to $z$,
            then add $L(T) \cap \descendant(w) \setminus \{z\}$ to $B_{\ell_i(v)}$.
    \end{itemize}
    Observe that for each tuple the number of vertices added to $B_s$ is equal to the increase of $b_s$ in definition of the ILP, hence, we have $|B_s|=b_s$ for each $s \in \flex$.
    Moreover, at this point the number of unassigned players of type $t \in \flex$ is equal to $c_t$ because we started with all non-\template vertices and assigned exactly one for each tuple in $K_{\mathcal{T}}$ which reflects the decrease of $c_t$ by one in the definition of the ILP.

    For each $s,t \in \flex$ assign $x_{s,t}$ players of type $t$ to $B_s$.
    We assign $x_{s,t}$ players of $\tau^{-1}(t)$ to the leaves $B_s$.
    Since $x_{s,t}$ is a solution to $I_{\mathcal{T}}$ we know by \cref{const:flex} that we assign the right number of players to each bag, by \cref{const:flex-tuples} we know that from each type we assign the remaining unassigned $c_t$ players.
    Lastly by \cref{const:cannot-beat} we know that the assigned players will lose as appropriate to players assigned to vertices of the \template (for tuples in $J_{\mathcal{T}}$) or $z$ (for tuples in $K_{\mathcal{T}}$).
    Hence, a feasible ILP implies we have a yes-instance.
\end{proof}

\begin{proof}[Proof of \Cref{thm:SimulTFPisFPT}]
    Given an instance of \STF, $(D_1, \dots, D_m, \winner)$, we first calculate the types and the \types-digraphs generated by each $D_i$. Then,
    according to \Cref{lem:number-of-templates}, we can iterate through the feasible \templates in FPT time.
    For each of these \templates we prepare an ILP instance.
    We first initialize it in $\bigoh(k^2)$ time.
    Then for each \template we
    calculate the $\bigoh(n)$ elements of $J_{\mathcal{T}}$ and $K_{\mathcal{T}}$ and we perform $\bigoh(1)$ operations on each element.
    The ILP contains \Cref{const:flex,const:flex-tuples} $|\flex|$ times each.
    Each contains $|\flex|$ variables and a constant. It also contains \Cref{const:cannot-beat} $|\flex| \cdot (|\flex|-1)$ times and \Cref{const:nonnegative} $|\flex|^2$ times each with one variable and one constant.
    All values are upper bounded by $n$ so the total number of bits in the input is $\bigoh(|\flex|^2 \cdot \log n)$ and we have $|\flex|^2$ variables.
    As $|\flex| = 2k+1$ by \Cref{prop:lenstra} we solve the ILP instance in $\bigoh\big((k^2)^{2.5 k^2+o(k^2)}\cdot k^2 \log n\big)$ time and space polynomial in $k^2 \cdot \log n$.
    If the ILP instance is feasible then, by \Cref{lem:feasible_ILP_has_STF}, we answer yes.
    If we have iterated through every \template and none of the ILP instances were feasible we answer no by \Cref{lem:feasible_stf_has_ilp}.
\end{proof}
    
\section{Application to \PTF}\label{sec:applications}


Let us define the parameterization we consider for \PTF. 
Recall an instance of \PTF is of the form $(N,P,p^*,\alpha^*)$ where $P$ is a matrix of pairwise winning probabilities over the player set $N$ and $p^*$ is the target probability with which we want the player $\alpha^*$ to win.

\begin{definition2}[Parameter for \PTF]
  Recall that for an instance $(N,P,p^*,\alpha^*)$ of \PTF, the \emph{certainty digraph} is the digraph defined over the vertex set $N$ where there is an arc $uv$ in the digraph if and only if $P_{u,v}=1$.
Define the \emph{degree of uncertainty} of this instance as the number of pair sets $\{u,v\}$ from $N$ such that $P_{u,v}$ is not equal to 0 or 1.
\end{definition2}

Notice that in the above definition, if $(N,P,p^*,\alpha^*)$ is an instance of \PTF, then the degree of uncertainty is half the number of fractional values in $P$.

As a consequence of Theorem \ref{thm:SimulTFPisFPT}, we obtain the following algorithm for \PTF.
\begin{theorem}\label{thm:mainProbTFPisFPT}
   \PTF is \FPT parameterized by the \FAS number of the certainty digraph and the degree of uncertainty.
\end{theorem}

We observe that the shared digraph in \STF is analogous to the certainty digraph in \PTF. Hence, we prove \Cref{thm:mainProbTFPisFPT} by proving the following lemma that reduces 	\PTF to \STF and then using the algorithm of Theorem~\ref{thm:SimulTFPisFPT} in the premise of Lemma~\ref{lem:STS-implies-PTF}.

\begin{lemma}\label{lem:STS-implies-PTF}
   If one can solve \STF in time $\mathcal T$, then \PTF can be solved in time $\mathcal T \cdot 2^{2^k} \cdot n^{\bigoh(1)}$ where $k$ is the degree of uncertainty of the \PTF instance and $n$ is the input size.
\end{lemma}
\begin{proof}
   Let $(N,P,\winner,p^*)$ denote an instance of PTF and let $C$ be the certainty digraph. 
   
   Let us first sketch the idea behind the reduction. If we were to run the probabilistic experiment (thereby determining who wins each uncertain match) we would get a new arc between every pair of players $u,v$ where neither $uv$ nor $vu$ are in $C$: the arc $uv$ appears in the tournament with probability $P_{uv}$, otherwise the arc $vu$ appears there instead.
   Our goal in the reduction is to consider all $2^k$ possible outcomes of this random process. They are tournament digraphs on $N$ (i.e. ``completions of the certainty digraph''), call them $D_1, \dots, D_{2^k}$, where $C \subseteq D_i$ for each $i\in [2^k]$.
   Note that, since the orientations of the arcs not in the certainty digraph are chosen independently, the probability that we would get $D_i$ is just the product of each $P_{u,v}$ where $uv \in E(D_i) \setminus E(C)$.
Since the  tournament digraphs $D_i$ are elementary events in our sample space, we have that the probability of an event $\mathcal D \subseteq \{D_1,\dots,D_{2^k}\}$ occurring is simply the sum of the probabilities of each $D_i \in \mathcal D$. So we can, completely deterministically, for each event $\mathcal D$, calculate its probability of occurrence.

Let us now return to the actual reduction. As described above, for each event $\mathcal D$, we calculate its probability of occurrence. If it is at least $p^*$ we create a new \STF instance comprising the player $\alpha^{*}$ and the digraphs in $\mathcal D$ and solve the instance in $\mathcal T$ time using the algorithm assumed in the premise.

   If any of these instances is a yes-instance (i.e. there is a seeding $\gamma$ where $\winner$ wins the bracket generated by $\gamma$ with respect to $D_i \in \mathcal D$ for every $i$), then we argue that this seeding is a solution for \PTF as follows.
   The probability of at least one of the $D_i$s occurring is at least $p^*$ and in each of them \winner wins the bracket generated by $\gamma$. Hence the probability of \winner winning the bracket generated by $\gamma$ is at least $p^*$.

Conversely, if we have a solution seeding for the PTF instance then there is some non-empty collection of $D_i$s ($i\in [2^k]$) such that $\alpha^*$ wins in each. Call this collection $\cal D$. Since we started with a solution seeding, the probability of $\cal D$ occurring (i.e., the sum of the probabilities of occurrence of the $D_i$s in $\cal D$) is at least $p^*$ and hence this collection is one of the STF instances created in our reduction.
 
    As there are $2^{2^k}$ possible events $\mathcal D$ and for each, we perform a polynomial-time processing to construct the \STF instance and then invoke the assumed algorithm that runs in $\mathcal T$ time, we obtain $\mathcal T \cdot 2^{2^k} \cdot n^{\bigoh(1)}$ time complexity for PTF. 
\end{proof}


\section{Concluding remarks and future work}

We have obtained the first fixed-parameter tractability results for SE tournament design with potential imperfect information. The rich body of work on the deterministic version provides natural directions for future research: for example, a probabilistic version of parameterization with respect to feedback vertex set number, as studied by Zehavi~\cite{Zehavi23} would be an improvement over this work. Alternately, probabilistic modeling of {\it demand tournament fixing} or {\it popular tournament fixing}, where the goal is to schedule  "high-value" matches as opposed to ensuring a specific player wins, as studied by Gupta et al.~\cite{GuptaSS24} and Chaudhary et al.~\cite{ChaudharyMZ24}, respectively, are also possible research directions.  

\newpage
\bibliographystyle{plain}
\bibliography{Aamas.bib}

\end{document}